\documentclass[intlimits,twoside,a4paper]{article}
\usepackage{amsthm}
\usepackage[cp1251]{inputenc}
\usepackage{xcolor}
\usepackage[eqsecnum]{cmpj3}

\usepackage{bm}

\theoremstyle{plain} 
\newtheorem{theorem}{Theorem}[section]
\newtheorem{proposition}{Proposition}[section]

\theoremstyle{definition}
\newtheorem{definition}{Definition}[section]

\theoremstyle{remark}
\newtheorem{remark}{Remark}[section]
\issue{2020}{23}{2}{23502}
\doinumber{10.5488/CMP.23.23502}

\title[A phase transition in a Curie-Weiss system]{A phase transition in a Curie-Weiss system with binary interactions\footnote{This paper is dedicated to Professor Ihor Mryglod on the occasion of his 60th birthday and in recognition of his significant  contribution to statistical physics, in particular, to the elaboration of new methods of statistical theory of fluids based on the collective variables method and investigation of liquid-gas critical phenomena in simple and multicomponent systems. }}

\author[Yu.V. Kozitsky,	M.P. Kozlovskii, O.A. Dobush] {Yu.V. Kozitsky\refaddr{label1},
	M.P. Kozlovskii\refaddr{label2}, O.A. Dobush\refaddr{label2}}
\addresses{
	\addr{label1} Institute of Mathematics, Maria Curie-Sk{\l}odowska University, Plac Marii Curie-Sk{\l}odowskiej 1, \\
	Lublin 20-031, Poland
	\addr{label2} Institute for Condensed Matter Physics of the National
	Academy of Sciences of Ukraine,\\ 1 Svientsitskii St., 79011 Lviv,
	Ukraine 
}
\date{Received January 30, 2020, in final form March 30, 2020}

\newcommand{\be}{\begin{equation}}
\newcommand{\ee}{\end{equation}}
\newcommand{\bea}{\begin{eqnarray}}
\newcommand{\eea}{\end{eqnarray}}

\usepackage{dsfont}
\usepackage{graphicx}
\usepackage{epstopdf}

\begin{document}

\maketitle

\begin{abstract}
	
A single-sort continuum Curie-Weiss system of interacting particles is studied. The particles are placed in the space $\mathds{R}^d$ divided into congruent cubic cells. For a region $V\subset
\mathds{R}^d$ consisting of $N\in \mathds{N}$ cells, every two particles contained in $V$ attract each other with intensity $J_1/N$. The particles contained in the same cell are subjected to
binary repulsion with intensity  $J_2>J_1$. For fixed values of the
temperature, the interaction intensities, and the chemical potential the thermodynamic phase is defined as a probability measure on the space of occupation numbers of cells, determined by a condition typical of Curie-Weiss theories. It is proved that the half-plane $J_1\,\times\,$\textit{chemical potential} contains phase coexistence
points at which there exist two thermodynamic phases of the system. An equation of state for this system is obtained.


%
%\pacs 51.30.+i, 64.60.fd, Mathematics Subject Classification 2000: 82B26, 82B21 
\end{abstract}

\section{Introduction}
\label{sec:1}

The mathematical theory of phase transitions in continuum particle
systems has much fewer results as compared to its counterpart
dealing with discrete underlying sets like lattices, graphs, etc. It
is then quite natural that the first steps in such theories are
being made by employing various mean field models. In \cite{LebP},
the mean field approach was mathematically realized by using a
Kac-like infinite range attraction combined with a two-body
repulsion. By means of rigorous upper and lower bounds for the
canonical partition function obtained in that paper, the authors
derived the equation of state indicating the possibility of a
first-order phase transition. Later on, this result was employed in
\cite{Lebow,presutti} to go beyond the mean field, see  also
\cite{pulvirenti,PT} for recent results. Another way of realizing
the mean-field approach is to use Curie-Weiss interactions and then
appropriate the methods of calculating the asymptotics of integrals, cf.
\cite{EN}. Quite recently, this way was formulated as a coherent
mathematical theory based on the large deviation techniques, in the
framework of which the Gibbs states (thermodynamic phases) of the
system are constructed as probability measures on an appropriate
phase space, see \cite[section~II]{Opoku}.

In this work, we introduce a simple Curie-Weiss type model of a
single-sort continuum particle system in which the space
$\mathds{R}^d$ is divided into congruent (cubic) cells. For a
bounded region $V\subset\mathds{R}^d$ consisting of $N$ such cells,
the attraction between every two particles in $V$ is set to be
$J_1/N$, regardless of their positions. If such two particles lie in
the same cell, they repel each other with intensity $J_2>J_1$.
Unlike  \cite{LebP}, we deal with the grand canonical ensemble.
Therefore, our initial thermodynamic variables are the inverse
temperature $\beta = 1/k_\text{B} T$ and the physical chemical potential.
However, for the sake of convenience we employ the variables
$p=\beta J_1$ and $\mu = \beta\,\times\,$(\textit{physical chemical
	potential}) and define single-phase domains of the half-plane
$\{(p,\mu): p>0, \mu \in \mathds{R}\}$ (see definition \ref{1df}) by
a condition that ensures the existence of a unique $\bar{y}\in
\mathds{R}$, which determines a probability measure
$\mathbf{Q}_{p,\mu}$, given in (\ref{19a}) and (\ref{19}). In the
grand canonical formalism and in the approach of \cite{Opoku}, this
measure is set to be the thermodynamic phase of the system. The
points $(p,\mu)$ where the mentioned single-phase condition fails to
hold due to the existence of multiple $\bar{y}$ correspond to the
coexistence of multiple thermodynamic phases. In theorem~\ref{4pn},
we show that there exists $p_0>0$ such that $\mathcal{R}(p_0):=\{(p,
\mu): p\in (0,p_0), \mu \in \mathds{R}\}$ is a single-phase domain,
that is, there is no phase-coexistence point in the strip
$\mathcal{R}(p_0)$ for small enough attractions. Note that
for some models on graphs, see  the example given in
\cite{KK}, there exist multiple phases for each positive attraction.
Thereafter, in theorem \ref{1tm} we show that for each value of
$J_2/J_1:=a>1$, there exist points in which two phases do coexist.
Namely,  we show that there exists $p_1>0$ such that for each
$p\geqslant p_1$, there exists $\mu_\text{c} (p)\in \mathds{R}$ and small enough
$\epsilon >0$ such that the sets $\{(p,\mu): \mu\in ( \mu_\text{c} (p)-
\epsilon, \mu_\text{c} (p))\}$ and $\{(p,\mu): \mu\in (\mu_\text{c}(p), \mu_\text{c} (p)
+ \epsilon)\}$ lie in different single-phase domains, and that there
exist at least two $\bar{y}$ whenever $\mu = \mu_\text{c}(p)$. In section~\ref{sec:3}, we present a number of results of the corresponding numerical
calculations which illustrate the postulates proved in theorems \ref{4pn}
and \ref{1tm}.

\section{The model}
\label{sec:2}

By $\mathds{N}$, $\mathds{R}$ and $\mathds{C}$ we denote the sets of
natural, real and complex numbers, respectively. We also put
$\mathds{N}_0 = \mathds{N}\cup \{0\}$. For $d \in \mathds{N}$, by
$\mathds{R}^d$, we denote the Euclidean space of vectors $x=(x^1,
\dots, x^d)$, $x^i\in \mathds{R}$. In the sequel, its dimension $d$
will be fixed. By $ d x$ we mean the Lebesgue measure on
$\mathds{R}^d$.
In this section,
we use some tools of the analysis on configuration spaces whose main
aspects  can be found in \cite{Albev}.

\subsection{The integration}
\label{ssec:2.1} For $n\in \mathds{N}$, let $\Gamma^{(n)}$ be the
set of all $n$-point subsets of $\mathds{R}^d$. Every subset of this kind,
$\gamma\in \Gamma^{(n)}$ (called \emph{configuration}), is  an
$n$-element set of distinct points $x\in \mathds{R}^d$. Let also
$\Gamma^{(0)}$ be the one--element set consisting of the empty
configuration. Every $\Gamma^{(n)}$ is equipped with the topology
related to the Euclidean topology of $\mathds{R}^d$. Then, we define
\begin{equation*}
% \label{1}
\Gamma_0 = \bigsqcup_{n \in \mathds{N}_0} \Gamma^{(n)},
\end{equation*}
that is, $\Gamma_0$ is the topological sum of the spaces
$\Gamma^{(n)}$. We equip $\Gamma_0$ with the corresponding Borel
$\sigma$-field $\mathcal{B}(\Gamma_0)$ that makes
$(\Gamma_0,\mathcal{B}(\Gamma_0))$ a standard Borel space. A
function $G: \Gamma_0 \to \mathds{R}$ is
$\mathcal{B}(\Gamma_0)$-measurable if and only if, for each $n\in
\mathds{N}$, there exists a symmetric Borel function
$G^{(n)}:(\mathds{R}^d)^n \to \mathds{R}$ such that $G(\gamma) =
G^{(n)}(x_1 , \dots, x_n)$ for $\gamma = \{x_ 1, \dots , x_n\}$. For
such a function $G$, we also set $G^{(0)}= G(\varnothing)$. The
Lebesgue-Poisson measure $\lambda$ on
$(\Gamma_0,\mathcal{B}(\Gamma_0))$ is defined by the relation
\begin{equation}
\label{2}
\int_{\Gamma_0} G (\gamma) \lambda ( \rd \gamma) = G^{(0)} +
\sum_{n=1}^\infty \frac{1}{n!} \int_{(\mathds{R}^d)^n} G^{(n)} ( x_1
, \dots , x_n) \rd x_1 \dots \rd x_n\,,
\end{equation}
which should hold for all measurable $G:\Gamma_0 \to \mathds{R}_+
:=[0, +\infty)$ for which the right-hand side of (\ref{2}) is
finite.

For some $c>0$, we let $\Delta= (-c/2, c/2]^d\subset \mathds{R}^d$
be a cubic cell of volume $\upsilon= c^d$ centered at the origin.
Let also $V\subset \mathds{R}^d$ be the union of $N\in \mathds{N}$
disjoint translates $\Delta_\ell$ of $\Delta$, i.e.,
\begin{equation*}
V = \bigcup_{\ell =1}^N \Delta_\ell.
% \label{3}
\end{equation*}
As is usual for Curie-Weiss theories, cf. \cite{EN,Opoku}, the form
of the interaction energy of the system of particles placed in $V$
depends on $V$. In our model, the energy of a configuration $\gamma
\subset V$ is
\begin{eqnarray}
\label{4} W_N (\gamma) & = & \frac{1}{2} \sum_{x, y \in \gamma}
\Phi_N
(x, y), \nonumber \\[.2cm]
\Phi_N (x, y) & = & - J_1/N + J_2 \sum_{\ell =1}^N
\mathds{I}_{\Delta_\ell}(x) \mathds{I}_{\Delta_\ell}(y),
\end{eqnarray}
where $\mathds{I}_{\Delta_\ell}$ is the indicator of $\Delta_\ell$,
that is, $\mathds{I}_{\Delta_\ell}(x)=1$ if $x\in \Delta_\ell$ and
$\mathds{I}_{\Delta_\ell}(x) =0$ otherwise. For convenience, in
$W_N$ above we have included the self-interaction term
$\Phi_N(x,x)$, which does not affect the physics of the model. We
also write $W_N$ and $\Phi_N$ instead of $W_V$ and $\Phi_V$ since
these quantities depend only on the number of cells in $V$. The
first term in $\Phi_N$ with $J_1>0$ describes the attraction. By virtue
of the Curie-Weiss approach, it is taken equal for all particles.
The second term with $J_2>0$ describes the repulsion between two
particles contained in one and the same cell. That is, in our model
every two particles in $V$ attract one another independently of their
location, and repel if they are in the same cell. The intensities
$J_1$ and $J_2$ are assumed to satisfy the following condition
\begin{equation}
\label{5}
J_2 > J_1.
\end{equation}
The latter is to secure the stability of the interaction
\cite{Ruelle}, that is, to satisfy
\[
\int_V \Phi_N (x, y) \rd y >0, \qquad {\rm for} \ {\rm all} \ x \in V.
\]
Let $\beta=1/k_{\rm B}T$ be the inverse temperature. To optimize the
choice of the thermodynamic variables we introduce the following
ones
\begin{equation}
\label{7a}
p = \beta J_1\,, \qquad a = J_2/J_1\,,
\end{equation}
and the dimensionless chemical potential $\mu= \beta\,
\times\,$(physical chemical potential). Then,  $(p, \mu)\in
\mathds{R}_{+}\times \mathds{R}$ is considered to be the basic set of
thermodynamic variables, whereas $a$ and $\upsilon$ are model
parameters.

The grand canonical partition function in region $V$ is
\begin{eqnarray}
\label{6}
\Xi_N (p, \mu)& = & \int_{\Gamma_V} \exp\left[  \mu |\gamma| - \beta W_N (\gamma)
\right]  \lambda ( \rd \gamma) \nonumber\\
& = & \int_{\Gamma_V} \exp\left[  \mu |\gamma| +
\frac{p}{2N}|\gamma|^2 - \frac{ap}{2}\sum_{x,y\in \gamma}
\sum_{\ell=1}^N \mathds{I}_{\Delta_\ell}(x)\mathds{I}_{\Delta_\ell}(y) \right]  \lambda ( \rd \gamma), 
\end{eqnarray}
where $|\gamma|$ stands for the number of points in the configuration
$\gamma$, and $\Gamma_V$ is the subset of $\Gamma_0$ consisting of
all $\gamma$ contained in $V$. We write $\Xi_N$ instead of $\Xi_V$
for the  reasons mentioned above.

\subsection{Transforming the partition function}

Now we use a concrete form of the energy as in (\ref{4}) to bring
(\ref{6}) to a more convenient form. For a given $\ell =1, \dots ,
N$ and a configuration $\gamma\in \Gamma_V$, we set $\gamma_\ell =
\gamma \cap \Delta_\ell$, that is, $\gamma_\ell$ is the part of the
configuration contained in $\Delta_\ell$. Then, $|\gamma_\ell|$ will
stand for the number of points of $\gamma$ contained in
$\Delta_\ell$. Note that
\begin{equation}
\label{6a}
|\gamma_\ell| = \sum_{x\in \gamma_\ell} 1 = \sum_{x\in \gamma}
\mathds{I}_{\Delta_\ell}(x).
\end{equation}
Then, cf. (\ref{4}) and (\ref{6}),
\begin{eqnarray*}
	%\label{7}
	\sum_{x,y\in \gamma}  \Phi_N (x, y) = \sum_{\ell, \ell'
		=1}^N \sum_{x\in \Delta_\ell} \sum_{y \in \Delta_{\ell'}} \Phi_N (x,
	y) 
	 = - \frac{J_1}{N} \left( \sum_{\ell =1}^N
	|\gamma_\ell|\right)^2 + J_2 \sum_{\ell =1}^N \left\vert \gamma_\ell
	\right\vert^2. \nonumber
\end{eqnarray*}
To rewrite the integrand in (\ref{6}) in a more convenient form we
set
\begin{equation}
\label{7b}
F_N(\varrho, p, \mu) = \exp\left[ \frac{p}{2N}
\left(\sum_{\ell=1}^N \varrho_\ell \right)^2 +  \mu \sum_{\ell=1}^N
\varrho_\ell - \frac{a p }{2}\sum_{\ell=1}^N \varrho_\ell^2
\right] ,
\end{equation}
where $\varrho \in \mathds{N}_0^N$ is a vector with nonnegative
integer components $\varrho_\ell$, $\ell = 1, 2, \dots , N$. Then,
(\ref{6}) takes the form
\begin{equation}
\label{8}
\Xi_N( p, \mu) = \int_{\Gamma_V} F_N \big(\nu(\gamma), p, \mu\big) \lambda ( \rd
\gamma),
\end{equation}
where $F_N$ is as in (\ref{7b}) and $\nu(\gamma)\in \mathds{N}_0^N$
is the vector with component $|\gamma_\ell|$, $\ell = 1, \dots , N$.
For $n, m \in \mathds{N}_0$, the Kronecker $\delta$-symbol can be
written
\begin{equation*}
%\label{9}
\delta_{nm} = \int_0^1 \exp \left[  2 \piup {\rm i}t (n-m)\right] \rd t,
\qquad {\rm i} = \sqrt{-1}.
\end{equation*}
Applying this in (\ref{8}) we get
\begin{eqnarray}
\Xi_N(p, \mu) &  = &   \sum_{\varrho \in \mathds{N}_0^N}
F_N(\varrho, p, \mu)  
\int_{\Gamma_V} \int_{[0,1]^N} \exp\left[ 
2 \piup {\rm i} \sum_{\ell=1}^N \left( \varrho_\ell - |\gamma_\ell|
\right)t_\ell
\right]  \lambda ( \rd \gamma) \rd t_1 \dots \rd t_N \nonumber \\
& = & \sum_{\varrho \in \mathds{N}_0^N} F_N(\varrho, p, \mu)
 \int_{[0,1]^N} \exp\bigg{(} 2 \piup {\rm i}
\sum_{\ell=1}^N \varrho_\ell t_\ell \bigg{)} R_N (t_1 , \dots , t_N)
\rd t_1 \dots \rd t_N. 
\label{10}
\end{eqnarray}
Here,
\begin{eqnarray}
\label{10a} R_N (t_1 , \dots , t_N) & = & \int_{\Gamma_V}
\exp\bigg{(} - 2 \piup {\rm i} \sum_{\ell=1}^N |\gamma_\ell| t_\ell
\bigg{)} \lambda ( \rd \gamma)
%\nonumber \\[.2cm] 
 = \int_{\Gamma_V}  \exp\bigg{[} - 2 \piup {\rm
	i} \sum_{\ell=1}^N \sum_{x\in \gamma} \mathds{I}_{\Delta_\ell} (x)
t_\ell \bigg{]}
\lambda ( \rd \gamma) \nonumber\\[.2cm]
& = &\sum_{n=0}^\infty \frac{1}{n!} \int_{V^n} \exp\bigg[  - 2 \piup
{\rm i} \sum_{\ell=1}^N \sum_{j=1}^n \mathds{I}_{\Delta_\ell}(x_j)
t_\ell \bigg] \rd x_1 \dots \rd x_n. 
\end{eqnarray}
In getting the second line of (\ref{10a}), we use (\ref{6a}), and
then the integral with $\lambda$ is written according to~(\ref{2}).
Note that the expression under the integral in the last line of
(\ref{10a}) factors in $j$, which allows for writing it in the form
\begin{eqnarray*}
	R_N (t_1 , \dots , t_N)& = & \sum_{n=0}^\infty \frac{1}{n!}\left\lbrace  
	\int_{V} \exp\left[  - 2 \piup {\rm i} \sum_{\ell=1}^N
	\mathds{I}_{\Delta_\ell}(x) t_\ell \right] 
	\rd x\right\rbrace^n \\[.2cm]
	& = & \sum_{n=0}^\infty \frac{1}{n!} \left[ \sum_{\ell=1}^N
	\int_{\Delta_\ell} \exp\left( - 2 \piup {\rm i} t_\ell \right) \rd x
	\right]^n 
	 =  \exp \left[  \upsilon \sum_{\ell =1}^N \exp\left(  -
	2 \piup {\rm i} t_\ell \right)  \right]  .
\end{eqnarray*}

Now we apply this in (\ref{10}) and obtain
\begin{eqnarray}
\label{11}
\Xi_N (p, \mu)  =  \sum_{\varrho\in \mathds{N}_0^N} F_N(\varrho, p, \mu)
\prod_{\ell=1}^N \left(
\frac{\upsilon^{\varrho_\ell}}{\varrho_\ell!}\right) 
= 
\sum_{\varrho\in \mathds{N}_0^N} \exp\left[  \frac{p}{2N} \left(
\sum_{\ell=1}^N \varrho_\ell\right)^2\right] \prod_{\ell=1}^N \pi
(\varrho_\ell, \mu),
\end{eqnarray}
where $p$ is as in (\ref{7a}) and
\begin{equation}
\label{12}
\pi (n, \mu) = \frac{\upsilon^n}{n!} \exp\left( \mu n- \frac{1}{2}
a p n^2\right), \qquad n \in \mathds{N}_0.
\end{equation}
Note that, for $p=0, \pi$ turns into the (non-normalized) Poisson
distribution with parameter $\upsilon \re^{\,\mu}$. Hence, alternating the
cell size amounts to shifting $\mu $.

\subsection{Single-phase domains}

By a standard identity
\[
\exp\left[  \frac{p}{2N} \left(
\sum_{\ell=1}^N \varrho_\ell\right)^2\right] = \sqrt{\frac{N}{2\piup
		p}} \int_{\mathds{R}}\exp\Bigg{(} - N \frac{y^2}{2p} + y
\sum_{\ell=1}^N \varrho_\ell\Bigg{)} \rd y,
\]
we transform (\ref{11}) into the following expression
\begin{equation}
\label{13}
\Xi_N (p , \mu) = \sqrt{\frac{N}{2\piup
		p}} \int_{\mathds{R}}\exp\left[  N E (y, p , \mu) \right] \rd y,
\end{equation}
where
\begin{equation}
\label{14}
E (y, p , \mu) = - \frac{y^2}{2p} + \ln K (y, p , \mu),
\end{equation}
and, cf. (\ref{7a}) and (\ref{12}),
\begin{equation}
\label{15}
K (y, p , \mu) = \sum_{n=0}^\infty \frac{\upsilon^n}{n!}
\exp\left[ (y+\mu)n - \frac{a p}{2} n^2 \right].
\end{equation}
Note that $E$ is an infinitely differentiable function of all its
arguments. Set
\begin{equation}
\label{15a}
P_N (p , \mu) = \frac{1}{\upsilon N} \ln \Xi_N (p , \mu).
\end{equation}
By the following evident inequality
\[
(y +  \mu) n - \frac{ap}{2}n^2 \leqslant \frac{(y +  \mu)^2}{2 ap} \,,
\qquad n\in \mathds{N}_0\,,
\]
we obtain from (\ref{15}) and  (\ref{14}) that
\begin{equation}
\label{16}
E (y, p, \mu) \leqslant - \frac{a-1}{2ap} y^2 + \frac{ \mu}{2ap} (2y +
\mu) + \upsilon.
\end{equation}
By virtue of Laplace's method \cite{Fed}, to calculate the large $N$
limit in (\ref{15a})  we should  find the global maxima of $E (y, p,
\mu)$ as a function of $y\in \mathds{R}$.
\begin{remark}
	\label{1rk}
	From the estimate in (\ref{16}) it follows that: (a) the integral in
	(\ref{13}) is convergent for all $p
	>0$ and $\mu \in \mathds{R}$ since $a>1$, see (\ref{5}) and
	(\ref{7a});  (b) for fixed $p$ and $\mu$,  as the bounded from the above
	function $E (y, p , \mu)$ has global maxima, each of which is also
	its local maximum.
\end{remark}
To get (b) we observe that (\ref{16}) implies $\lim_{|y|\to
	+\infty}E(y, p, \mu) = -\infty$; hence, each point $\bar{y}$ of
global maximum belongs to a certain interval $(\bar{y}-\varepsilon,
\bar{y}+\varepsilon)$, where it is also a maximum point. Since $E$
is everywhere differentiable in $y$, then $\bar{y}$ is the point of
global maximum only if it solves the following equation
\begin{equation}
\label{16b}
E_1 (y, p, \mu) := \frac{\partial}{\partial y} E(y, p, \mu) =0.
\end{equation}
By (\ref{14}) and (\ref{15}) this equation can be rewritten in the
form
\begin{eqnarray}
\label{16c}
&& -\frac{y}{p} + \frac{K_1 (y, p, \mu)}{K (y, p, \mu)} = 0, \\
&& K_1 (y, p, \mu):= \sum_{n=1}^\infty \frac{n\upsilon^n}{n!}
\exp\left[(y+\mu)n - \frac{ap}{2} n^2 \right]. \nonumber
\end{eqnarray}
\begin{remark}
	\label{RSrkk}
	As we will see from the proof of theorem \ref{4pn} below, the
	equation in (\ref{16c}) has at least one solution for all $p>0$ and
	$\mu \in \mathds{R}$. Since both $K_1$ and $K$ take only strictly
	positive values, these solutions are also strictly positive.
\end{remark}
\begin{definition}
	\label{1df}
	We say that  $(p, \mu)$ belongs to a single-phase domain if $E (y, p
	, \mu)$ has a unique global maximum $\bar{y}\in\mathds{R}$ such that
	\begin{equation}
	\label{16a}
	E_2 (\bar{y}, p , \mu):=\frac{\partial^2}{\partial y^2} E (y, p ,
	\mu)\Big|_{y=\bar{y}} < 0.
	\end{equation}
\end{definition}
Note that $\bar{y}$ can be a point of maximum if $E_1 (\bar{y}, p ,
\mu)= E_2 (\bar{y}, p , \mu)=0$. That is, not every point of global
maximum corresponds to a point in a single-phase domain.

The condition in (\ref{16c})  determines the unique probability
measure $Q_{p,\mu}$ on $\mathds{N}_0$ such that
\begin{equation}
\label{19}
Q_{p,\mu} (n) = \frac{1}{K(\bar{y}, p , \mu)n!} \upsilon^n
\exp\left[ (\bar{y}+ \mu)n -\frac{ap}{2}n^2 \right], \qquad n \in
\mathds{N}_0\,,
\end{equation}
which yields the probability law of the occupation number of a
single cell. Then, the unique thermodynamic phase of the model
corresponding to $(p, \mu)\in \mathcal{R}$ is the product
\begin{equation}
\label{19a}
\mathbf{Q}_{p, \mu} = \bigotimes_{\ell=1}^\infty Q^{(\ell)}_{p, \mu}
\end{equation}
of the copies of the measure defined in (\ref{19}). It is a
probability measure on the space of all vectors $\mathbf{n} =
(n_\ell)_{\ell=1}^\infty$, in which $n_\ell\in \mathds{N}_0$ is the
occupation number of $\ell$-th cell.

The role of the condition in (\ref{16a}) is to yield the possibility
to apply Laplace's method for asymptotic calculation of the integral in
(\ref{13}). By direct calculations, it follows that
\begin{eqnarray}
\label{17}
& & E_2 (y, p , \mu) = - \frac{1}{p}  + \frac{1}{2\left[K (y,
	p , \mu)\right]^2}\nonumber\\[.2cm]  & & \quad \times \sum_{n_1, n_2 =0}^\infty
\frac{\upsilon^{n_1+n_2}}{n_1! n_2!} (n_1 - n_2)^2 \exp\left[  (y+
\mu)(n_1+n_2) - \frac{ap}{2}( n_1^2 + n_2^2) \right].
\end{eqnarray}
In dealing with the equation in (\ref{16c}) we  fix $p>0$ and
consider $E_1$ as a function of $y\in \mathds{R}$ and $\mu\in
\mathds{R}$. Then, for a given $\mu_0$, we solve (\ref{16c}) to find
$\bar{y}_0$ and then check whether it is the unique point of global
maximum and (\ref{16a}) is satisfied, i.e., whether $(p,\mu_0)$
belongs to a single-phase domain. Then, we slightly vary $\mu$ and
repeat the same. This will yield a function $\mu \mapsto
\bar{y}(\mu)$ defined in the neighbourhood of $\mu_0$, which dependends
on the choice of $p$ and satisfies $\bar{y}(\mu_0) =\bar{y}_0$. In
doing so, we use the analytic implicit function theorem based
on the fact that, for each fixed $p>0$, the function
$\mathds{R}^2\ni (y, \mu) \mapsto E_1(y,p,\mu)$ can be analytically
continued to some complex neighbourhood of $\mathds{R}^2$, see
(\ref{15}) and (\ref{16c}). For the reader's convenience, we present
this theorem here in the form adapted from 
\cite[section 7.6, page 34]{CC}.
 For some $p_0>0$, let $\mathcal{B}\subset \mathds{C}^2$ be
a connected open set containing $\mathds{R}^2$ such that the
function $(y, \mu) \mapsto E_1(y,p_0,\mu)$ should be analytic in
$\mathcal{B}$.
\begin{proposition}[Implicit function theorem]
	\label{0pn}
	Let $p_0$ and $(y_0,\mu_0)$  be such that $E_1 (y_0,p_0,\mu_0) =0$
	and $E_2 (y_0,p_0,\mu_0) \neq 0$. Let also $\mathcal{B}\subset
	\mathds{C}^2$ be as just described. Then, there exist open sets
	$\mathcal{D}_1\subset \mathds{C}$ and $\mathcal{D}_2\subset
	\mathds{C}$ such that $y_0 \in \mathcal{D}_1$, $ \mu_0 \in
	\mathcal{D}_2$, $\mathcal{D}_1\times \mathcal{D}_2 \subset
	\mathcal{B}$, and an analytic function $\bar{y}: \mathcal{D}_2 \to
	\mathcal{D}_1$, for which the following holds
	\begin{equation*}
	%\label{17a}
	\{ (y,\mu)\in \mathcal{D}_1 \times \mathcal{D}_2 : E_1 (y,p_0,\mu) =
	0\} = \{ (\bar{y}(\mu),  \mu): \mu\in \mathcal{D}_2\}.
	\end{equation*}
	The derivative of $\bar{y}$ in $\mathcal{D}_2$ is
	\begin{eqnarray}
	\label{SR1}
	\frac{\rd \bar{y} (\mu)}{\rd \mu} =  - \frac{1}{E_2 (\bar{y}(\mu), p_0,
		\mu)} \left[\frac{\partial }{\partial \mu} E_1 (y, p_0, \mu)\right]_{y=\bar{y}(\mu)}.
	\end{eqnarray}
\end{proposition}
\begin{remark}
	\label{Luftrk}
	In the sequel, we  also use the version of the implicit function
	theorem in which we do not employ the analytic continuation of $E_1$
	to complex values of $p$. Let the conditions of proposition~\ref{0pn} regarding $E_1$ and $E_2$ be satisfied. Then, there exist
	open sets $\mathcal{D}_i \subset \mathds{R}$, $i=1,2,3$, and a
	continuous function $\bar{y}:\mathcal{D}_3 \times \mathcal{D}_2 \to
	\mathcal{D}_1$ such that $p_0 \in \mathcal{D}_3$, $\mu_0 \in
	\mathcal{D}_2$ and the following holds
	\begin{equation*}
	%\label{177}
	\{ (y,p,\mu)\in \mathcal{D}_1 \times \mathcal{D}_2
	\times\mathcal{D}_3 : E_1 (y,p_0,\mu) = 0\} = \{ (\bar{y}(p,\mu), p,
	\mu): p\in \mathcal{D}_3, \ \mu\in \mathcal{D}_2\}.
	\end{equation*}
\end{remark}

The partial derivative of $\bar{y} (p, \mu)$ over $\mu\in
\mathcal{D}_2$ is given by the right-hand side of (\ref{SR1}). In the
sequel, by writing $\bar{y} (\mu)$ we assume both the function as
proposition \ref{0pn}, defined for a fixed $p$ known from the
context, and that as in remark \ref{Luftrk} with the fixed value of
$p$.

For a fixed $p_0>0$, assume that $(p_0, \mu_0)$ belongs to a
single-phase domain. By proposition \ref{0pn} there exists
$\varepsilon >0$ such that the function $(\mu_0 - \varepsilon, \mu_0
+ \varepsilon)\ni \mu \mapsto \bar{y}(\mu)$ can be defined by the
equation $E_1 (y, p_0, \mu)=0$. Its continuation from the mentioned
interval is related to the fulfilment of the condition $E_2
(\bar{y}(\mu), p_0 , \mu) <0$, cf. (\ref{16a}), which may not be the
case. At the same time, by (\ref{17}) we have that
\begin{equation*}
%\label{172}
\frac{\partial }{\partial \mu} E_1 (y, p, \mu)= E_2 (y, p, \mu) +
\frac{1}{p}
>0,
\end{equation*}
holding for all $y\in\mathds{R}$, $p>0$ and $\mu \in\mathds{R}$. In
view of this and  (\ref{SR1}), it might be more convenient to use
the inverse function $y\mapsto \bar{\mu}(y)$ since the
$\mu$-derivative of $E_1$ is always nonzero. Its properties are
described by the following statement obtained from the analytic
implicit function theorem mentioned above. Recall that only positive
$y$ solves the equation in (\ref{16c}).
\begin{proposition}
	\label{01pn}
	Given $p_0$, let $\mathcal{B}$ be as in proposition \ref{0pn}. Then,
	there exist open connected subsets $\mathcal{D}_i\subset
	\mathds{C}$, $i=1,2$, and an analytic function $\mathcal{D}_1\ni y
	\mapsto \bar{\mu}(y)\in \mathcal{D}_2$ such that $\mathcal{D}_1$
	contains $\mathds{R}_{+}$, $\mathcal{D}_1 \times \mathcal{D}_2
	\subset \mathcal{B}$, and the following holds
	\begin{equation*}
	%\label{170a}
	\{ (y,\mu)\in \mathcal{D}_1 \times \mathcal{D}_2 : E_1 (y,p_0,\mu) =
	0\} = \{ (y,  \bar{\mu}(y)): y\in \mathcal{D}_1\}.
	\end{equation*}
	The derivative of $\bar{\mu}$ in $\mathcal{D}_1$ is
	\begin{eqnarray*}
		%\label{SR2}
		\frac{\rd \bar{\mu} (y)}{\rd y}  =  - \frac{E_2 (y, p_0,
			\bar{ \mu}(y))}{E_2 (y, p_0, \bar{ \mu}(y)) + \frac{1}{p}}.
	\end{eqnarray*}
\end{proposition}
\begin{proposition}
	\label{1pn}
	Each single-phase domain, $\mathcal{R}$, has the following
	properties: (a) it is an open subset of $\mathds{R}_{+} \times
	\mathds{R}$; (b) for each $(p_0, \mu_0)\in \mathcal{R}$, the
	function $\mathcal{I}_{p_0}:=\{ \mu \in \mathds{R}: (p_0, \mu )\in
	\mathcal{R}\}\ni \mu \mapsto \bar{y}(\mu)$ as in proposition
	\ref{0pn} is continuously differentiable on $\mathcal{I}_{p_0}$.
	Moreover,
	\begin{equation}
	\label{Luft}
	\frac{\rd\bar{y}(\mu)}{\rd \mu}  >0 , \qquad {\textit for} \ {\textit all} \
	\mu\in \mathcal{I}_{p_0}.
	\end{equation}
\end{proposition}
\begin{proof}
	For a single-phase domain, $\mathcal{R}$, take $(p_0, \mu_0)\in
	\mathcal{R}$. By remark \ref{Luftrk} the function $(p,\mu) \mapsto
	\bar{y}(p,\mu)$, defined by the equation $E_1(y,p,\mu)=0$ is
	continuous in some open subset of $\mathcal{R}$ containing $(p_0,
	\mu_0)$. By the continuity of $E_2 (\bar{y}(p,\mu), p, \mu)$ and the
	fact that $E_2 (\bar{y}(p_0,\mu_0), p_0, \mu_0)<0$ (since $(p_0,
	\mu_0)\in \mathcal{R}$) we get that $E_2 (\bar{y}(p,\mu), p, \mu)<0$
	for $(p,\mu)$ in some open neighbourhood of $(p_0, \mu_0)$. Hence,
	$\mathcal{R}$ contains $(p_0, \mu_0)$ with some neighbourhood and
	thus is open. The continuous differentiability of $\bar{y}$ and the
	sign rule in (\ref{Luft}) follow by proposition \ref{0pn} and
	(\ref{SR1}), respectively.
\end{proof}

\vspace{-3mm}
By (\ref{19}) and  (\ref{16c}) we get the $Q_{p,\mu}$-mean value
$\bar{n} = \bar{n}(p,\mu)$ of the occupation number of a given cell
in the form
\begin{equation}
\label{oct}
\bar{n}(p, \mu) = \sum_{n=0}^\infty n Q_{p,\mu}(n) =
\frac{K_1(\bar{y}(p,\mu),p,\mu)}{K(\bar{y}(p,\mu),p,\mu)}
=  \frac{\bar{y}(p , \mu)}{p}.
\end{equation}
Note that, up to the factor $\upsilon^{-1}$, $\bar{n}(p, \mu)$ is
the particle density in phase $\mathbf{Q}_{p, \mu}$. For a fixed
$p$, by proposition~\ref{1pn} $\bar{n}(p, \cdot)$ in an increasing
function on $\mathcal{I}_p$, which  can be inverted to give
$\bar{\mu}(p, \bar{n})$. By Laplace's method we then get the
following corollary of proposition \ref{1pn}.
\begin{proposition}
	\label{2pn}
	Let $\mathcal{R}$ be a single-phase domain. Then, for each $(p ,
	\mu) \in \mathcal{R}$, the limiting pressure $P (p , \mu) =
	\lim_{N\to +\infty} P_N(p , \mu)$, see (\ref{15a}), exists and is
	continuously differentiable on $\mathcal{R}$. Moreover, it is given
	by the following formula
	\begin{equation}
	\label{21}
	P(p, \mu) = \upsilon^{-1} E(\bar{y}(p , \mu), p ,
	\mu) .
	\end{equation}
\end{proposition}
Let $\mathcal{N}_p$ be the image of $\mathcal{I}_p$ under the map
$\mu\mapsto \bar{n}(p, \mu)$. Then, the inverse map $\bar{n}\mapsto
\bar{\mu}(p , \bar{n})$ is continuously differential and increasing
on $\mathcal{N}_p$. By means of this map, for a fixed $p$, the
pressure given in (\ref{21}) can be written as a function of
$\bar{n}$
\begin{equation}
\label{22}
P = \bar{P}(\bar{n}) = \upsilon^{-1} E(p\bar{n}, p ,
\bar{ \mu}(p , \bar{n})), \qquad \bar{n}\in \mathcal{N}_p\,,
\end{equation}
which is the equation of state.

\subsection{The phase transition} \label{sec2.4}

Recall that the notion of the single-phase domain was introduced in
definition \ref{1df}, and each  domain of this kind is an open subset of
the open right half-plane $\{(p,\mu): p>0, \ \mu \in \mathds{R}\}$,
see proposition \ref{1pn}. With this regard we have the following
possibilities: (i) the whole half-plane $\{(p,\mu): p>0, \ \mu \in
\mathds{R}\}$ is such a domain; (ii) there exist more than one
single-phase domain. In case (i),  for all $(p,\mu)$ there exists
one phase (\ref{19a}). In the context of this work, a phase
transition is understood as the possibility of having different
phases at the same value of the pair $(p,\mu)$. If this is the case,
$(p,\mu)$ is called a \emph{phase coexistence point}. Clearly, such
a point should belong to the common topological boundary of at least
two distinct single-phase domains. That is, to prove the existence
of phase transitions we have to show that possibility (ii) takes
place and that these single-phase domains have a common boundary. We
do this in theorems~\ref{4pn} and \ref{1tm} below.

Let $\mathcal{R}$ be a single-phase domain. Take $(p_0,\mu_0)\in
\mathcal{R}$ and consider the line $l_{p_0}=\{(p_0,\mu): \mu\in
\mathds{R}\}$. If the  whole line lies in $\mathcal{R}$, by
proposition \ref{1pn}, $\bar{y}(\mu)$ is a continuously
differentiable and increasing function of $\mu \in \mathds{R}$. In
our first theorem, we prove that this is the case for small enough
$p_0$.
\begin{theorem}
	\label{4pn}
	There exists $p_0>0$ such that the set $\mathcal{R}(p_0):=\{ (p,
	\mu): p\in(0,p_0]\}$ is a single-phase domain.
\end{theorem}
\begin{proof}
	In view of remark \ref{1rk}, we have to show that, for fixed $p\leqslant
	p_0$ and all $\mu\in \mathds{R}$, $E(y,p,\mu)$ has exactly one local
	maximum such that (\ref{16a}) holds. For $x\in \mathds{R}$, we set,
	cf. (\ref{15}),
	\begin{eqnarray}
	\label{23}
	\phi(x,p) &=& \ln \sum_{n=0}^\infty \frac{\upsilon^n}{n!} \exp\left(
	x n - \frac{ap}{2} n^2\right), \\ \nonumber
	\phi_k (x,p) &=&
	\frac{\partial^k}{\partial x^k}\phi(x,p), \qquad k=1,2.
	\end{eqnarray}
	Similarly to (\ref{16}), we  get
	\begin{equation}
	\label{24}
	\phi(x,p) \leqslant \upsilon + \frac{x^2}{2ap}\,.
	\end{equation}
	Note also that, for  the functions defined in (\ref{23}), we have
	\begin{equation}
	\label{23a}
	\lim_{p\to 0} \phi(x,p) = \lim_{p\to 0} \phi_1(x,p) = \lim_{p\to 0}
	\phi_2(x,p) = \upsilon \re^x, \qquad x\in \mathds{R}.
	\end{equation}
	By means of (\ref{23}) we rewrite the equation in (\ref{16c}) in the
	following form
	\begin{equation}
	\label{25}
		\left\{
	\begin{array}{ll}
	 x = y + \mu,\\
	%[.3cm]
	 \mu = x -  p \phi_1 (x, p).
	\end{array}
	\right.
	\end{equation}
	Our aim is to show that there exists $p_0>0$ such that, for each
	$p\in (0,p_0]$, the following holds: (a) the second line in
	(\ref{25}) defines an increasing unbounded function $\tilde{\mu}
	(x)$, $x\in \mathds{R}$; (b) $p \phi_2 (x,p) \leqslant 1-\delta$ for some
	$\delta\in (0,1)$ and all $x\in \mathds{R}$. Indeed, the function
	mentioned in (a) can be inverted to give an unbounded increasing
	function $\tilde{x} (\mu)$, $\mu \in \mathds{R}$, such that the
	solution of (\ref{16c}) is $\bar{y}(\mu) = \tilde{x}(\mu ) - \mu$.
	Then, by (\ref{23}) and (\ref{17}) we get
	\[
	E_2 (\bar{y}(\mu), p, \mu) = -\frac{1}{p} + \phi_2 (\tilde{x}(\mu),
	p) <0,
	\]
	where the latter inequality follows by (b). Thus, to prove both (a)
	and (b), it is enough to show that there exists positive $p_0$ such
	that
	\begin{equation}
	\label{26}
	p \phi_2 (x,p) \leqslant \frac{1}{a}\,, \qquad {\rm for} \ \ {\rm all} \ \
	x \in \mathds{R} \ \ {\rm and} \ \ p\in (0,p_0].
	\end{equation}
	By (\ref{23}) we have
	\begin{eqnarray}
	\label{d} \phi_2 (x,p) & = & \frac{1}{2\Phi(x,p)} \sum_{n_1, n_2
		=0}^\infty
	\frac{\upsilon^{n_1+n_2}}{n_1! n_2!} (n_1-n_2)^2   \exp \left[  x
	(n_1 + n_2) - \frac{ap}{2} (n_1^2 + n_2^2)\right],  \\[.2cm]
	\label{d1} 
	\Phi(x,p) & = & \left[ \sum_{n=0}^\infty
	\frac{\upsilon^n}{n!} \exp\left(x n - \frac{ap}{2} n^2
	\right)\right]^2 \nonumber\\[.2cm]
	& = & \sum_{n_1, n_2 =0}^\infty \frac{\upsilon^{n_1+n_2}}{n_1! n_2!}
	\exp \left[  x (n_1 + n_2) - \frac{ap}{2} (n_1^2 + n_2^2)\right].
	\end{eqnarray}
	Since $\Phi(x,p) \geqslant 1$, we  get from the latter
	\begin{eqnarray}
	\label{27}
	\phi_2 (x,p) & \leqslant & \frac{1}{2} \sum_{n_1, n_2 =0}^\infty
	\frac{\upsilon^{n_1+n_2}}{n_1! n_2!} (n_1-n_2)^2 \exp \left[  x
	(n_1 + n_2) - \frac{ap}{2} (n_1^2 + n_2^2)\right]  \nonumber \\[.2cm]
	\nonumber & \leqslant & \frac{1}{2} \sum_{n_1, n_2 =0}^\infty
	\frac{\upsilon^{n_1+n_2}}{n_1! n_2!} (n_1-n_2)^2 \exp \left[  x
	(n_1 + n_2)\right]  \\[.2cm]  & = &
	\left[ \sum_{n=1}^\infty \frac{n^2}{n!} \left(\upsilon \re^x
	\right)^n \right]  \cdot \left[ \sum_{n=0}^\infty \frac{1}{n!}
	\left(\upsilon \re^x \right)^n \right]  - \left[ \sum_{n=1}^\infty
	\frac{n}{n!} \left(\upsilon \re^x \right)^n \right]^2 =
	\upsilon \exp ( x+ 2 \upsilon \re^x ).
	\end{eqnarray}
	Fix some $x_0>0$ and then set
	\begin{equation}
	\label{27a}
	p_{01} = \frac{1}{a\upsilon} \exp (- x_0 - 2 \upsilon \re^{x_0} ).
	\end{equation}
	Then, by (\ref{27}) we get
	\begin{equation}
	\label{29a}
	p \phi_2 (x,p) \leqslant \frac{1}{a}\,, \qquad {\rm for} \ \ {\rm all} \ \
	x \leqslant x_0 \ \ {\rm and} \ \ p\leqslant p_{01}.
	\end{equation}
	By (\ref{23}) and (\ref{23a}) we see that the function
	\begin{equation}
	\label{28}
	\psi(p) = \frac{\upsilon - \phi(0,p)}{\phi_1(0,p)}
	\end{equation}
	continuously depends on $p>0$ and $\psi(p) \to 0$ as $p\to 0$. For
	each $x>0$, one finds $\xi \in (0,1]$, dependent on $x$ and $p$,
	such that
	\begin{equation*}
	%\label{29}
	\phi (x,p) = \phi(0,p) + \phi_1 (0,p) x + \frac{1}{2} \phi_2(\xi
	x,p) x^2.
	\end{equation*}
	From this and (\ref{24}) we get
	\begin{equation}
	\label{29}
	\phi (\xi^{-1} x,p) = \phi(0,p) + \phi_1 (0,p) \xi^{-1} x +
	\frac{1}{2} \phi_2( x,p)\xi^{-2} x^2 \leqslant \upsilon + \frac{\xi^{-2}
		x^2}{2ap}\,.
	\end{equation}
	For the function defined in (\ref{28}) and $x_0$ as in (\ref{27a}),
	we pick $p_{02}$ such that $\psi(p) \leqslant x_0$ for all $p\leqslant
	p_{02}$. For such values of $p$, this yields
	\[
	\phi_1(0,p) \xi^{-1} x \geqslant \phi_1(0,p)  x_0  \geqslant \upsilon -
	\phi(0,p), \qquad {\rm for}\ \ {\rm all} \ \ x\geqslant x_0\,.
	\]
	We apply this in (\ref{29}) and get
	\begin{equation}
	\label{30}
	p \phi_2 (x, p) \leqslant \frac{1}{a}\,, \qquad {\rm for} \ \ {\rm all} \ \
	x \geqslant x_0 \ \ {\rm and} \ \ p\leqslant p_{02}\,.
	\end{equation}
	This and (\ref{29a})  yields (\ref{26}) with $p_0 =
	\min\{p_{01}; p_{02}\}$, which completes the proof.
\end{proof}
\begin{theorem}
	\label{1tm}
	For each $a>1$, there exists $p_1 = p_1 (a) >0$ such that, for each
	$p\geqslant p_1$, the line $l_{p}=\{(p, \mu):\mu \in \mathds{R}\}$
	contains at least one phase-coexistence point.
\end{theorem}
\begin{proof}
	Let $\bar{\mu}(y)$ be the function as in proposition \ref{01pn}. By
	(\ref{23}) and (\ref{25}) we have that
	\begin{equation}
	\label{s}
	\bar{\mu}(y) = \left[x - p \phi_1 (x,p) \right]_{x=\bar{x}(y)},
	\end{equation}
	and
	\begin{equation*}
	%\label{Naha}
	\lim_{y\to 0} \bar{\mu} (y) = - \infty, \qquad \lim_{y\to +\infty}
	\bar{\mu} (y) = + \infty.
	\end{equation*}
	In (\ref{s}), $\bar{x}(y)$ is the inverse of the function
	$\mathds{R}\ni x \mapsto y= p \phi_1 (x,p)$. Note that
	\begin{equation}
	\label{Naha1}
	\frac{\rd \bar{x}(y)}{\rd y} = \left[ \frac{1}{ p \phi_2
		(x,p)}\right]_{x=\hat{x}(y)},
	\end{equation}
	hence $(0, +\infty) \ni y \mapsto \bar{x}(y) \in \mathds{R}$ is
	increasing. By (\ref{s}) and (\ref{Naha1}) it follows that
	\begin{equation}
	\label{s1}
	\frac{\rd \bar{\mu}(y)}{\rd y} = \frac{\rd \bar{x}(y)}{ \rd y} \left[1 -
	p \phi_2 (x,p) \right]_{x=\bar{x}(y)} =  \left[ \frac{1}{ p \phi_2
		(x,p)}\right]_{x=\bar{x}(y)} - 1.
	\end{equation}
	For a given $p>0$, pick $x^p>0$ such that $\psi(p)$ defined in
	(\ref{28}) satisfies $\psi(p) \leqslant x^p$. Then, as in (\ref{30}) we
	obtain $p \phi_2 (x, p) \leqslant 1/a$ for all $x\geqslant x^p$. By (\ref{s1}),
	this yields that
	\begin{equation}
	\label{Naha2}
	\frac{\rd \bar{\mu}(y)}{\rd y} \geqslant a -1 , \qquad {\rm for} \ \ y \in
	[y^p, +\infty),
	\end{equation}
	where  $y^p= p\phi_1(x^p,p)$. For the same $p$, let $x_p$ be such
	that $a\upsilon p = \exp( - x_p -2 \upsilon \re^{x_p})$,
	cf. (\ref{27a}). Then, by (\ref{27}) and (\ref{s1}), we conclude that
	the inequality in (\ref{Naha2}) holds also for  $y\in (0,y_p]$,
	$y_p:= p\phi_1(x_p,p)$. As we have seen in the proof of theorem
	\ref{4pn}, the mentioned two intervals may overlap, i.e., it may be
	that $y_p > y^p$, if $p$ is small enough. Let us prove that this is
	not the case for big $p$. That is, let us show that there exists
	$p_1>0$ such that, cf. (\ref{s1}) and (\ref{30}), for all $p\geqslant
	p_1$, there exists $x\in \mathds{R}$ such that the following holds
	\begin{equation}
	\label{s2}
	p_1 \phi_2 (x, p_1) \geqslant 1 , \qquad {\rm and} \quad p \phi_2 (x, p)
	> 1, \quad {\rm for} \quad p>p_1.
	\end{equation}
	To this end, we estimate the denominator of (\ref{d}) from the above and
	the numerator from the below. For $x = a p/2$, by (\ref{d1})  it follows
	that
	\begin{eqnarray}
	\label{s3}
	\Phi (x, p) =  \left\lbrace  \sum_{n=0}^\infty \frac{\upsilon^{n}}{n!}
	\exp\left[  - \frac{ap}{2}n(n-1) \right] \right\rbrace ^2 \leqslant \left(
	\sum_{n=0}^\infty \frac{\upsilon^{n}}{n!}\right)^2 = \re^{2\upsilon},
	\end{eqnarray}
	where we used the fact that $p>0$. In the sum in the numerator of
	(\ref{d}), we take just two summands corresponding to $n_1=1$,
	$n_2=0$ and $n_1=0$, $n_2=1$, and obtain, by (\ref{s3}), the following
	estimate
	\begin{equation}
	\label{s4}
	p \phi_2 (x, p) \geqslant p \upsilon \re^{- 2 \upsilon},
	\end{equation}
	holding for $x=ap/2$.  Then, we set $p_1= \upsilon^{-1}\re^{2
		\upsilon}$. For this $p_1$ and $x =  ap_1/2$, by (\ref{s4}) we
	obtain (\ref{s2}). Clearly,  for $p> p_1$, $x_p$ and $x^p$
	introduced above satisfy
	\[
	x_p < ap/2 < x^p.
	\]
	For $p>p_1$, let $({x}_p, {x}^p)$ be the biggest interval which
	contains $x=ap/2$ and is such that $p \phi_2 (x, p) >1$ for each
	$x\in ({x}_p, {x}^p)$. Then, $p \phi_2 (x, p) =1$ for $x= {x}_p$ and
	$x ={x}^p$. Set
	\[
	{y}_p =p \phi_1 ({x}_p, p), \qquad {y}^p =p \phi_1 ({x}^p, p).
	\]
	Then, by (\ref{s1}) and (\ref{Naha2}) it follows that
	\begin{eqnarray}
	\label{s5}
	\frac{\rd\bar{\mu}(y)}{\rd y} & = & 0, \qquad {\rm for} \quad
	y = {y}_p , \ {y}^p, \nonumber\\[.2cm] \nonumber
	\frac{\rd\bar{\mu}(y)}{\rd y} & <  & 0, \qquad {\rm for} \quad
	y\in ({y}_p , {y}^p),
	\\[.2cm] 
	\frac{\rd\bar{\mu}(y)}{\rd y} & \geqslant & a-1> 0, \qquad {\rm for} \quad
	y< {y}_p  \ \ {\rm and} \ \ y> {y}^p.
	\end{eqnarray}
	From this we see that ${y}_p$ (resp. ${y}^p$) is the first maximum
	(resp. the last minimum) of $\bar{\mu}(y)$. Let $\hat{y}^p$ be the
	first minimum of $\bar{\mu}(y)$. Set $\hat{\mu}^p = \bar{\mu}
	(\hat{y}^p)$. Now we pick $y_2 > \hat{y}^p$ such that: (a) either
	$\bar{\mu}(y_2) = \bar{\mu}(y_p)$; (b) or $y_2$ is the second
	maximum of $\bar{\mu}(y)$ if $\bar{\mu}(y_2) \leqslant \bar{\mu}(y_p)$.
	Then, set $\hat{\mu}_p = \bar{\mu}(y_2)$. Clearly, $\hat{\mu}_p >
	\hat{\mu}^p$. In case (a), we have $\hat{\mu}_p = \bar{\mu}(y_p)$;
	and  $\hat{\mu}_p \leqslant \bar{\mu}(y_p)$ in case (b). Now we pick $y_1
	\in (0, y_p)$ such that $\bar{\mu}(y_1) = \bar{\mu}(\hat{y}^p)$. By
	(\ref{s5}) and the above construction,  the function $\bar{\mu}(y)$
	is increasing on $[y_1, y_p)$ and $(\hat{y}_p, y_2)$, and decreasing
	on $(y_p, \hat{y}^p)$, see figure \ref{fig_3}. Let $\bar{y}_1$ and
	$\bar{y}_2$ be the inverse functions to the restrictions of
	$\bar{\mu}(y)$ to $[y_1, y_p)$ and $(\hat{y}_p, y_2)$, respectively.
	Let also $\bar{y}_3$ be the inverse function to the restrictions of
	$\bar{\mu}(y)$ to the interval $(y_p, \hat{y}^p)$. All the three
	functions are defined on $M_p:=(\hat{\mu}^p, \hat{\mu}_p)$ and are
	continuously differentiable thereon.
	Note that
	$\bar{y}_2 (\mu) > \bar{y}_3 (\mu)
	> \bar{y}_1 (\mu)$ for all $\mu\in (\hat{\mu}^p, \hat{\mu}_p)$ and
	\begin{equation}
	\label{Nah}
	\bar{y}_1 (\hat{\mu}_p ) = \bar{y}_3 (\hat{\mu}_p ), \qquad
	\bar{y}_2 (\hat{\mu}^p ) = \bar{y}_3 (\hat{\mu}^p ).
	\end{equation}
	Moreover, all the three $\bar{y}_i (\mu)$, $i=1,2,3$, satisfy
	(\ref{16c}) and, for $\mu \in (\hat{\mu}^p, \hat{\mu}_p)$,
	$E(y,p,\mu)$ has local maxima at $y=\bar{y}_i(\mu)$, $i=1,2$ and a
	local minimum at $y=\bar{y}_3(\mu)$. This follows from the fact that
	$E_2 (\bar{y}_i (\mu), p, \mu)< 0$ for $\mu\in M_p$ and $i=1,2$, and
	from  $E_2 (\bar{y}_3 (\mu), p, \mu)> 0$ for $\mu\in M_p$, see
	(\ref{s5}).
	
	Set
	\begin{equation}
	\label{31} D(\mu) = E(\bar{y}_2 (\mu), p, \mu) - E(\bar{y}_1 (\mu), p,
	\mu), \qquad \mu \in M_p.
	\end{equation}
	If $D(\mu) < 0$, then $\bar{y}_1(\mu)$ is the point of global
	maximum of $E(y,p,\mu)$ and hence $(p,\mu)$ lies in a single-phase
	domain, say $\mathcal{R}_1$. If $D(\mu) > 0$, then the same holds
	for $\bar{y}_2(\mu)$ and $\mathcal{R}_2$. If
	\begin{equation}
	\label{Na}
	D(\mu_1) < 0, \qquad {\rm and} \qquad D(\mu_2) > 0,
	\end{equation}
	for some $\mu_1,\mu_2 \in M_p$, then there should exist $\mu_\text{c}$ in
	between where $D$ vanishes. Thus, $(p,\mu_\text{c})$ belongs to the
	boundaries of both $\mathcal{R}_1$ and $\mathcal{R}_2$, and hence is
	a phase coexistence point, if $\mu_\text{c}$ is an isolated zero of
	(\ref{31}). The phases are then given in (\ref{19}) and (\ref{19a})
	with $\bar{y}_1(\mu_\text{c})$ and $\bar{y}_2(\mu_\text{c})$, respectively. Note
	that the vanishing of $D$ at $\mu_\text{c}$ corresponds to the Maxwell
	rule, cf. \cite{LebP}, and to the existence of two global maxima of
	$E(y , p, \mu)$. Since both $\bar{y}_i(\mu)$ are differentiable, by
	(\ref{14}), (\ref{15}), (\ref{16b}), and (\ref{17}) we have
	\begin{eqnarray*}
		%\label{32}
		\frac{\rd D (\mu)}{\rd  \mu} & = & E_1 (\bar{y}_2 (\mu), p , \mu)
		\frac{\rd \bar{y}_2(\mu)}{ \rd \mu}  + \frac{K_1 (\bar{y}_2 (\mu), p ,
			\mu)}{K
			(\bar{y}_2 (\mu), p , \mu)} %\nonumber\\ 
		 -  E_1 (\bar{y}_1 (\mu), p , \mu) \frac{\rd
			\bar{y}_1(\mu)}{ \rd \mu}  - \frac{K_1 (\bar{y}_1 (\mu), p , \mu)}{K
			(\bar{y}_2 (\mu), p , \mu)}  \nonumber \\
		& = & \frac{K_1 (\bar{y}_2 (\mu), p ,\mu)} {K (\bar{y}_2 (\mu), p , \mu)} - \frac{K_1 (\bar{y}_1 (\mu), p
			, \mu)}{K (\bar{y}_2 (\mu), p , \mu)} 
		 =  E_2 (\bar{y}_*(\mu), p , \mu) + \frac{1}{p} >0
	\end{eqnarray*}
	for some $\bar{y}_* (\mu)\in [\bar{y}_1 (\mu), \bar{y}_2 (\mu)]$.
	Note that $E_1 (\bar{y}_i (\mu), p , \mu)=0$, $i=1,2$, cf.
	(\ref{16b}). Therefore, $D$ can hit the zero level once at most. Let
	us show that (\ref{Na}) does hold. If $\mu_1\in M_p$ is close enough
	to $\hat{\mu}^p$, then by (\ref{Nah}) and the mentioned continuity
	we have that $E(\bar{y}_2 (\mu_1), p, \mu_1)$ is close to
	$E(\bar{y}_3 (\mu_1), p, \mu_1)$, and hence $\bar{y}_2$ cannot be
	the global maximum of $E(y , p, \mu_1)$. Therefore, $D(\mu_1) < 0$
	for such $\mu$. Likewise, we establish the existence of $\mu_2$ such
	that $D(\mu_2) > 0$. Now, the existence of $\mu_\text{c}\in (\mu_1, \mu_2)$
	follows by the continuity and~(\ref{Na}).
\end{proof}

\section{Numerical results}
\label{sec:3}

Here, we present the results of numerical calculations of the
functions which appear in the preceding part of the paper.

We begin by considering the extremum points of the functions that
appear in section~\ref{sec2.4}. According to definition \ref{1df}, the line
$l_p =\{ (p,\mu): \mu\in \mathds{R}\}$ lies in a single-phase
domain, if the function $\mathds{R}_{+} \ni y \mapsto E(y, p, \mu)$,
see (\ref{14}), has a unique non-degenerate global maximum for all
$\mu\in \mathds{R}$. The corresponding condition in (\ref{16c})
determines an increasing function $\bar{y}(\mu)$, see proposition
\ref{1pn}, which can be inverted to give $\bar{\mu}(y)$, see
(\ref{s}). In theorem \ref{4pn}, we show that this holds for
 small enough $p$. Figure \ref{fig_1} presents the results of the
calculation of $\bar{\mu}(y)$ for

\begin{figure}[!b]
	\begin{centering}
		a)\includegraphics[width=150pt]{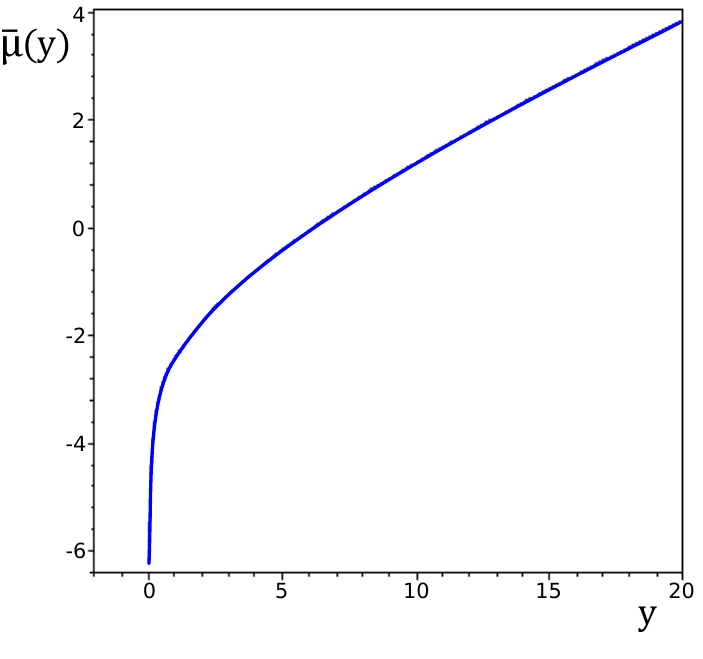} 
		b)\includegraphics[width=150pt]{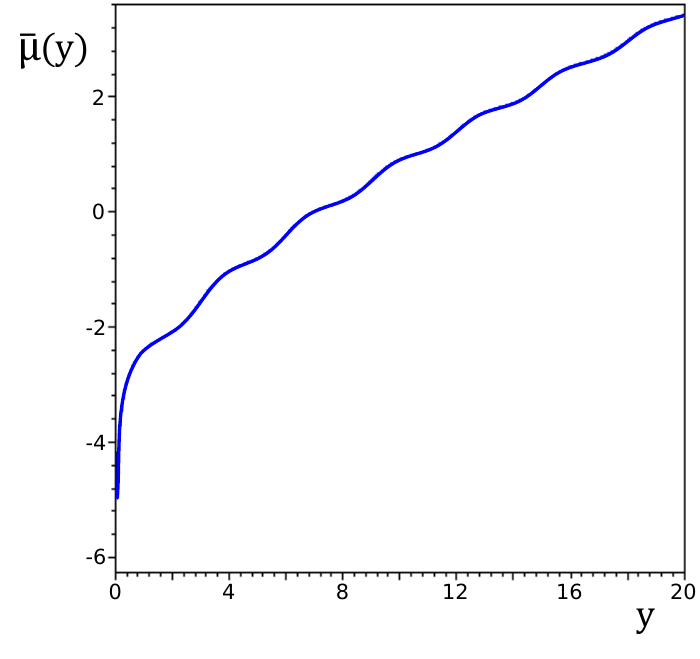}
		~~~~~~~~~~~~~~~~~~~~~~~~~~~~~~~~
		c)\includegraphics[width=150pt]{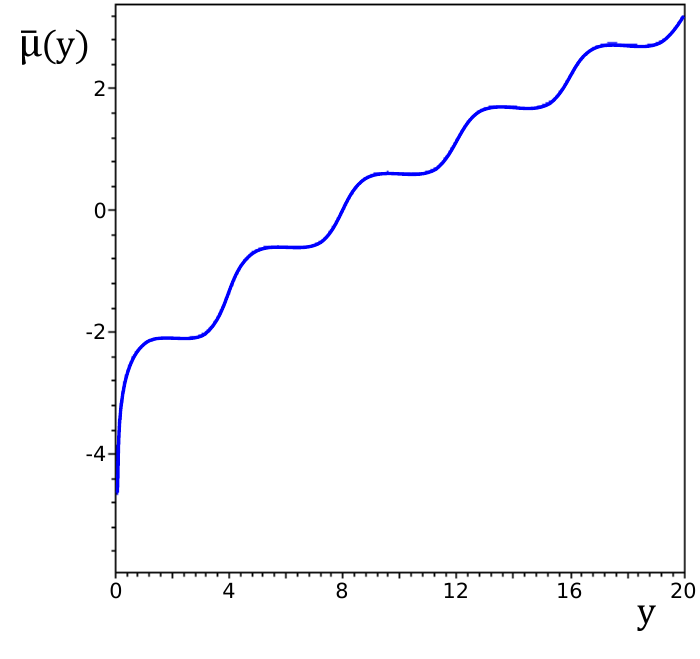}
		d)\includegraphics[width=150pt]{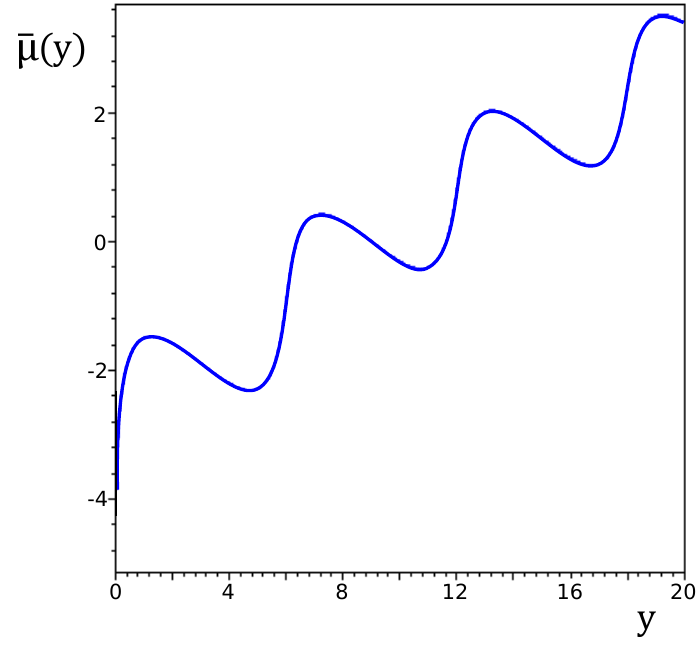}
		\caption{(Colour online) Plot of $\bar\mu(y)$ defined in (\ref{s}) for $a=1.2$,
			$\upsilon = 12$ and various values of the attraction parameter:
			$p=2$ (curve a), $p=3$ (curve b), $p=4$ (curve c), $p=6$ (curve d). }
		\label{fig_1}
	\end{centering}
\end{figure}
\begin{equation}\label{42m}
a = 1.2,\qquad \upsilon = 12,
\end{equation}
and $p=2,3,4,6$ --- curves (a), (b), (c) and (d), respectively. In
the first two cases, $\bar{\mu}$ is an increasing function, which
corresponds to the situation described in theorem \ref{4pn}. That
is, the values of $p=2$ and $3$ are below the critical value $p_\text{c} =
p_\text{c}(a)$. For $a$ and $\upsilon$ as in (\ref{42m}), our calculations
yield
\begin{equation*}\label{45m}
p_\text{c} = 3.928235(8).
\end{equation*}
For $p=p_\text{c}$, the function $y\mapsto E(y, p_\text{c}, \mu)$ still has a
unique global maximum, which gets degenerate, i.e., $E_2 (\bar{y},
p_\text{c}, \mu)=0$, cf. (\ref{16a}). The value of $\bar{y}=\bar{y}_\text{c}$ at
which this occurs gives the value of the critical density $\bar{n}_\text{c}
= \bar{y}_\text{c} / p_\text{c}$, see (\ref{22}). For various values of the
parameter $a$, see (\ref{7a}), the values of $p_\text{c}(a)$, $\bar{y}_\text{c}
(a)$ and $\bar{n}_\text{c} (a)$ are given in the following table~\ref{tab:table1}.
\begin{table}[!t]
	\centering
	\caption{Values of $p(a)$, $\bar y (a)$, $ \bar n (a)$ in the critical point for $v = 12$.}
	\label{tab:table1}
	\vspace{2ex}
	\begin{tabular}{|c|c|c|c|c|c|}
	\hline\hline
		$a$ \ \  & \  \ 1.0001 \ \ & \ \ 1.2 \ \ & \ \ 1.5 \ \ & \ \ 2 \  \ & \ \ 10 \ \ \\
		\hline
		$ p_\text{c}(a)$  &  3.8255  &   3.9282  &  3.9796  &  3.9973  &  4.0000  \\
		\hline
		$\bar y_\text{c} (a)$ \ \  & \  \ 2.0485 \ \ & \ \  2.0187  \ \ & \ \ 2.0052 \ \ & \ \ 2.0007 \  \ & \ \ 2.0000 \ \ \\
		\hline
		$\bar n_\text{c} (a)$ \ \  & \  \ 0.5355 \ \ & \ \  0.5139 \ \ & \ \ 0.5038 \ \ & \ \ 0.5005 \  \ & \ \ 0.5000 \ \ \\
		\hline\hline
	\end{tabular}
\end{table}

\begin{figure}[!t]
	\begin{centering}
		\includegraphics[width=0.4\textwidth]{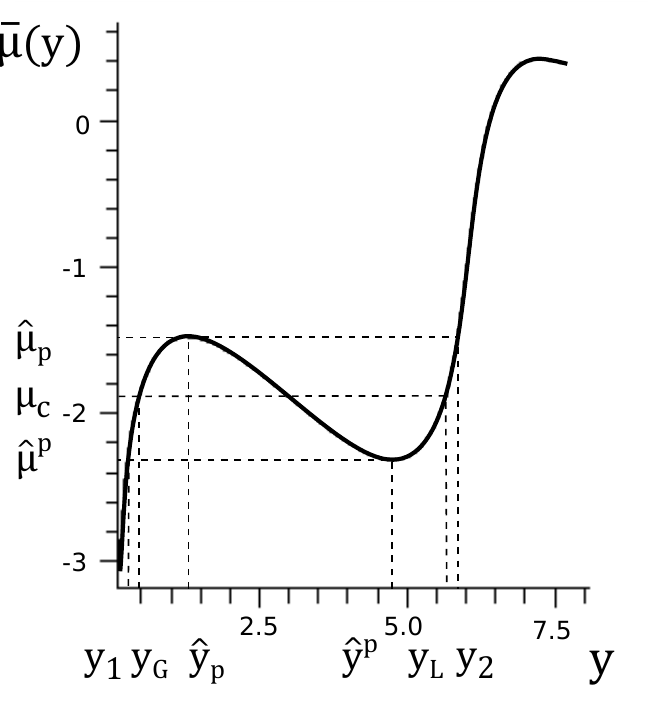}
		\caption{ Plot of $\bar{\mu}(y)$ for $a=1.2$ and $p=6$ more in 
			detail. } \label{fig_3}
	\end{centering}
\end{figure}
The values of $p=4$ and $6$ are above the critical point $p_\text{c}$,
which can be clearly  seen from the curves (c) and (d) of figure
\ref{fig_1}. In this case, one deals with the situation described by
theorem \ref{1tm}. figure \ref{fig_3} presents more in  detail the
curve plotted in figure \ref{fig_1} (d), i.e., corresponding to $p=6$
and $a=1.2$. It provides a good illustration to the proof of theorem
\ref{1tm}. Here, we have $\bar{\mu}(y_2) = \bar{\mu}(y_p)$.

Let us now turn to the maximum points of $E(y,p, \mu)$. For $a$ as
in (\ref{42m}) and $p=6$, figure \ref{fig_4} presents the dependence of
$E$ on $y$ for $\mu_1 = - 2.3080$ (curve a), and $\mu_2 = - 1.4700$
(curve b). This provides an illustration to (\ref{Na}). The curve
plotted in figure \ref{fig_5} corresponds to the critical value of
$\mu$ defined by the condition $D(\mu)=0$. That is, $(p,\mu_\text{c})$ with
$p=6$ and $\mu_\text{c}=-1.890291$ is the phase coexistence point whose
existence  was proved in theorem \ref{1tm}. Figure \ref{fig_6}
presents the dependence of $E(\bar{y}_1(\mu), p, \mu)$ (line 1, red)
and $E(\bar{y}_2(\mu), p, \mu)$ (line 2, blue) on $\mu\in M_p$,
$p=6$, cf. (\ref{31}). Their intersection occurs at
$\mu=\mu_\text{c}=-1.890291$.

\begin{figure}[!b]
	\begin{centering}
		a)\includegraphics[width=150pt]{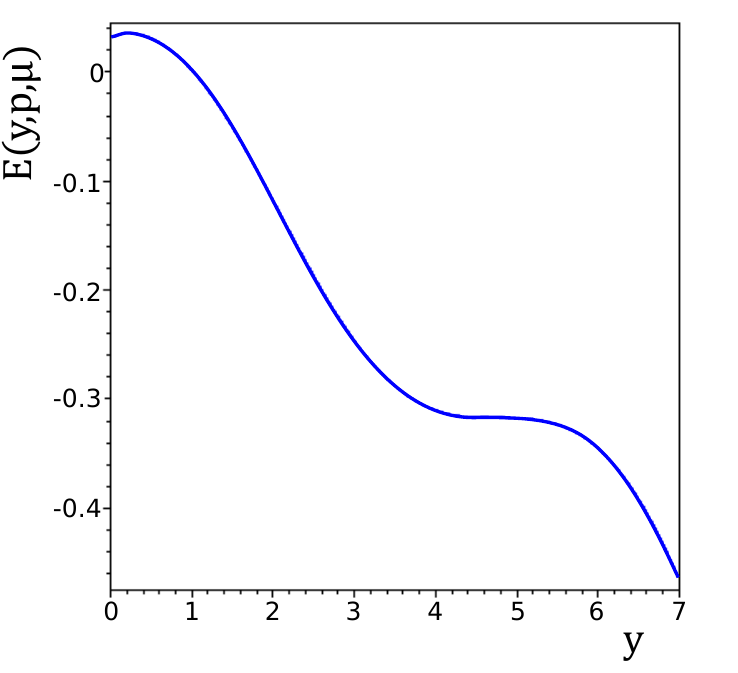} b)
		\includegraphics[width=150pt]{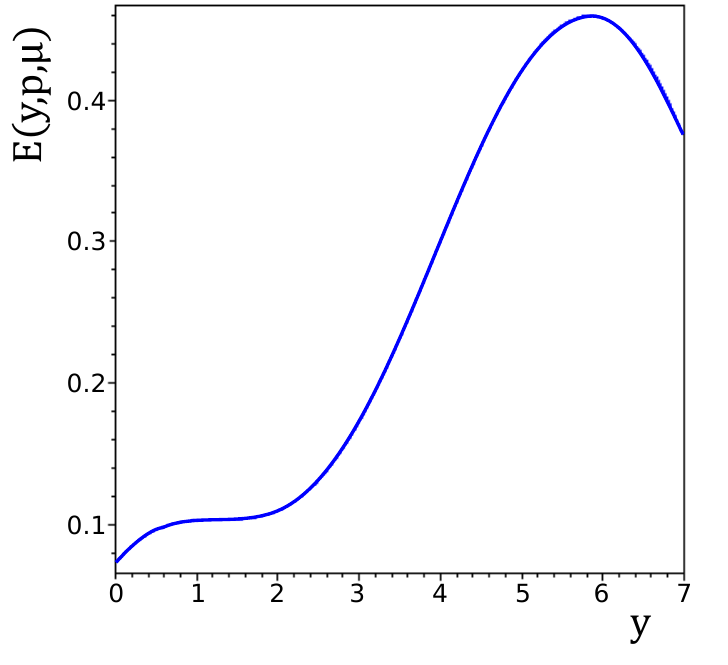}
		\caption{(Colour online)  Plot of the function $E(y, p,\mu)$ for $p=6$, $a=1.2$ and
			$\mu_1 = -2.3080$ (curve a), $\mu_2=-1.4700$ (curve b). }
		\label{fig_4}
	\end{centering}
\end{figure}

\begin{figure}[!t]
	\includegraphics[width=0.41\textwidth]{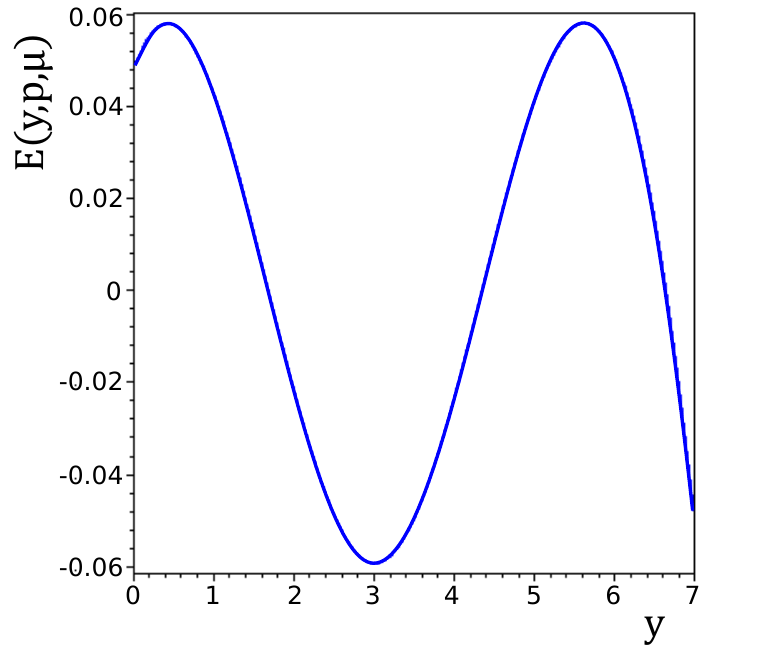}
	\hspace{1.5cm}
	\includegraphics[width=0.38\textwidth]{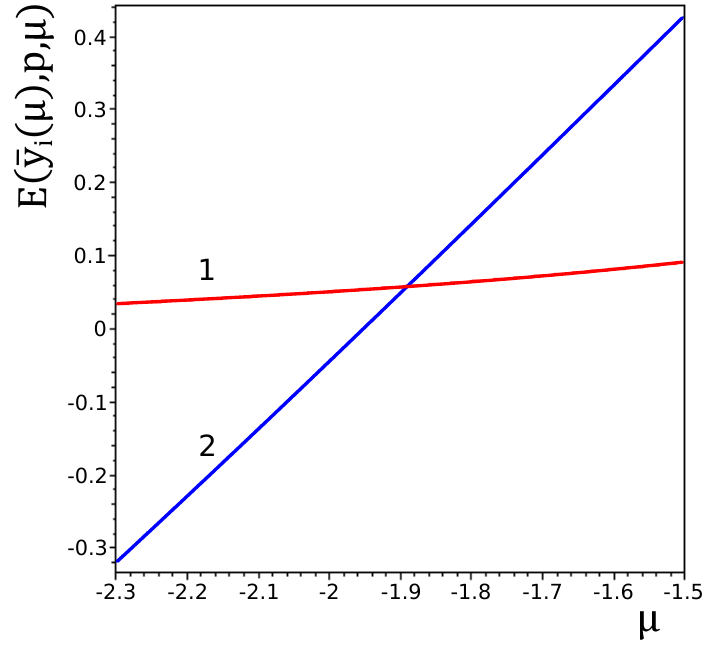}
	\\
	\parbox[t]{0.48\textwidth}{
		\vspace{-0.3cm}
		\caption{(Colour online) The same as in figure \ref{fig_4} for $\mu = \mu_\text{c} = - 1.890291$.}
		\label{fig_5}
	}
	\hfill
	\parbox[t]{0.48\textwidth}{
		\vspace{-0.3cm}
		\caption{(Colour online) Plot of the pressure as a function of temperature in the extremum points of the compressibility. Plot of the functions $M_p \ni \mu \mapsto E(y_i(\mu), p, \mu)$, $i=1,2$, see (\ref{31}), and $p=6$, $a =1.2$. }
		\label{fig_6}
	}
\end{figure}

\begin{figure}[!t]
	\begin{centering}
		\includegraphics[width=250pt]{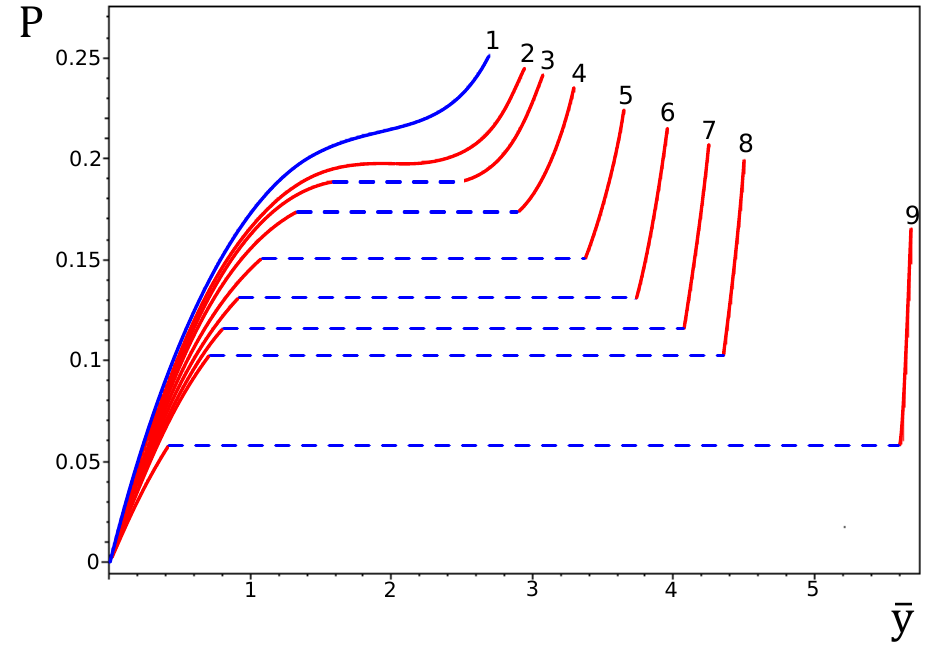}
		\caption{(Colour online) Plot of the dependence of the pressure on $\bar y = p
			\bar{n}$ (isotherms) see (\ref{21}) and (\ref{22}).  Curve 1
			corresponds to $p=3.8<p_\text{c}$. The curves 2--9 correspond to $p\geqslant
			p_\text{c}$:
			$p=p_\text{c}$  (curve 2),
			$p=4$ (curve 3),  $p=4.135$ (curve 4), $p=4.3647$  (curve 5),
			$p=4.5824$  (curve 6), $p=4.8$ (curve 7),  $p=5$ (curve 8), $p=6$
			(curve 9). } \label{fig_7}
	\end{centering}
\end{figure}

The curves plotted in figure \ref{fig_7} present the isotherms --- the
dependence of the pressure on $\bar{y}$, which is equivalent to the
dependence on the density $\bar{n}$, see (\ref{oct}), calculated
from (\ref{22}) for a number of fixed values of $p$.

\newpage

\section{Concluding remarks}

In this work, we proved the existence of multiple thermodynamic
phases at the same values of the extensive model parameters ---
temperature and chemical potential. In contrast to the approach of
\cite{LebP}, we deal directly with thermodynamic phases in the grand
canonical setting. To the best of our knowledge, this is the first
result of this kind.

\section*{Acknowledgements}

This work was supported in part by the European Commission (Seventh Framework Programme) under the project STREVCOMS PIRSES-2013-612669. Yuri Kozitsky was also supported by National Science Centre (NCN), Poland, grant 2017/25/B/ST1/00051, which is cordially acknowledged by him.

% ----------------------------------------------------------------

\ukrainianpart

\title{Фазовий перехід у системі із парною взаємодією Кюрі-Вейса}

\author[]{Ю.В. Козицький\refaddr{label1},
	М.П. Козловський\refaddr{label2}, О.А. Добуш\refaddr{label2}}
\addresses{
	\addr{label1} Інститут математики, Університет Марії Кюрі-Склодовської, \\пл. Марії Кюрі-Склодовської, 1, 20-031 Люблін, Польщa	\addr{label2} Інститут фізики конденсованих систем НАН України, вул. Свєнціцького, 1, 79011 Львів, Україна 
}

\makeukrtitle

\begin{abstract}
\tolerance=3000%

У роботі досліджено однокомпонентну неперервну систему частинок із взаємодією Кюрі-Вейса. Частинки перебувають у просторі  $\mathds{R}^d$, поділеному на одинакові кубічні комірки. Для області $V\subset
\mathds{R}^d$, що складається з $N\in \mathds{N}$ комірок, кожні дві частинки, що містяться в $ V $, притягують одна одну з інтенсивністю $ J_1 / N $. Частинки, що містяться в одній комірці, попарно відштовхуються з інтенсивністю $ J_2> J_1 $. Для фіксованих значень температури, інтенсивності взаємодії та хімічного потенціалу термодинамічна фаза визначається як міра ймовірності на просторі зайнятих чисел комірок, що визначається умовою, типовою для теорій Кюрі-Вейса. Доведено, що напівплощина $J_1\,\times\,$\textit{хімічний потенціал} містить точки фазового співіснування, при яких існують дві термодинамічні фази системи. Отримано рівняння стану для цієї системи. 

\keywords молекулярне поле, рівняння стану, співіснування фаз

\end{abstract}

\end{document}